\documentclass[11pt,a4paper]{article}

\usepackage[T1]{fontenc}
\usepackage[utf8]{inputenc}
\usepackage[english]{babel}
\usepackage[margin=1in]{geometry}
\usepackage[tbtags]{amsmath}
\usepackage{amsfonts,amssymb,amsthm,mathrsfs}
\usepackage[mathscr]{eucal}
\usepackage{commath}
\usepackage{booktabs}
\usepackage{graphicx}
\usepackage{xcolor}
\usepackage{aliascnt}
\usepackage[hidelinks]{hyperref}
\usepackage[capitalise,noabbrev]{cleveref}

\theoremstyle{plain}
\newtheorem{theorem}{Theorem}
\newaliascnt{lemma}{theorem}
\newtheorem{lemma}[lemma]{Lemma}
\aliascntresetthe{lemma}
\newaliascnt{corollary}{theorem}
\newtheorem{corollary}[corollary]{Corollary}
\aliascntresetthe{corollary}
\newaliascnt{proposition}{theorem}
\newtheorem{proposition}[proposition]{Proposition}
\aliascntresetthe{proposition}
\theoremstyle{definition}
\newaliascnt{remark}{theorem}
\newtheorem{remark}[remark]{Remark}
\aliascntresetthe{remark}
\newaliascnt{definition}{theorem}

\aliascntresetthe{definition}

\DeclareMathOperator{\RR}{\mathbb{R}}

\DeclareMathOperator{\QQ}{\mathbb{Q}}

\DeclareMathOperator{\ZZ}{\mathbb{Z}}

\DeclareMathOperator{\BB}{\mathbb{B}}

\DeclareMathOperator{\FF}{\mathbb{F}}

\DeclareMathOperator{\PC}{\mathcal{P}}

\DeclareMathOperator{\LC}{\mathcal{L}}

\DeclareMathOperator{\AC}{\mathcal{A}}
\DeclareMathOperator{\FC}{\mathcal{F}}

\DeclareMathOperator{\IC}{\mathcal{I}}

\DeclareMathOperator{\rank}{rank}

\DeclareMathOperator{\inputsize}{input\:size}
\DeclareMathOperator{\poly}{poly}

\DeclareMathOperator{\toddp}{td}

\DeclareMathOperator{\BUnit}{\mathbf 1}
\DeclareMathOperator{\BZero}{\mathbf 0}

\newcommand*{\intint}[2][1]{\left\{#1,\, \dots,\, #2\right\}}

\DeclareMathOperator{\NP}{N\!P}

\newcommand\restr[2]{{
  \left.\kern-\nulldelimiterspace 
  #1 
  \vphantom{\big|} 
  \right|_{#2} 
  }}

\title{Polynomial and Hyperplane Avoidance over the Integers with Applications to Integer-Point Counting}
\author{%
  Grigorii~Dakhno \and
  Dmitry~Gribanov \and
  Nikita~Kasianov \and
  Anastasiia~Kats \and
  Andrey~Kupavskii \and
  Nikita~Kuz'min \and
  Stanislav~Moiseev
}
\date{}

\hypersetup{
  bookmarksdepth=3,
  pdftitle={Polynomial and Hyperplane Avoidance over the Integers with Applications to Integer-Point Counting},
  pdfauthor={Grigorii Dakhno, Dmitry Gribanov, Nikita Kasianov, Anastasiia Kats, Andrey Kupavskii, Nikita Kuz'min, Stanislav Moiseev}
}

\begin{document}

\maketitle

\begin{abstract}
We study the \emph{Polynomial Avoidance Problem}: given a finite
collection $\FC\subseteq\ZZ[x_1,\dots,x_n]$ of nonzero polynomials,
find a vector $x\in\ZZ^n$ of minimum $\ell_p$-norm such that
$f(x)\neq0$ for every $f\in\FC$. Each polynomial is given sparsely as
a list of its nonzero terms. The affine special case is the
\emph{Hyperplane Avoidance Problem}. This special case arises in
Barvinok-type algorithms for counting integer points in polyhedra,
when certain rational generating functions must be evaluated at a
removable singularity. Our main results are as follows:
\begin{enumerate}
    \item Let
    $D:=\sum_{f\in\FC}\deg f$. Every instance admits a feasible vector
    $x$ satisfying
    \[
        \norm{x}_1\leq
        \min\left\{D,\frac{D+n}{2}\right\}.
    \]
    This bound is tight. Such a vector can be found by a
    deterministic algorithm whose bit complexity
    is polynomial in $D$ and the sparse encoding length of the input. Our proof uses a new algorithmic variant of
    the \emph{Combinatorial Nullstellensatz};

    \item In the affine case, with $m$ input polynomials, such a
    vector can be constructed by a fast deterministic polynomial-time algorithm using $O(nm)$ arithmetic operations. For the homogeneous case, Barvinok's
    original approach provides only the weaker guarantee
    $\norm{x}_1=O((nm)^{n-1})$;

    \item The problem is $\NP$-hard for every $p\in\RR_{\geq1}\cup\{\infty\}$, even for homogeneous linear forms;

    \item As an application, we improve the complexity bound for
    counting integer points in general rational polytopes. We also
    improve the bound for polyhedra of bounded codimension; this
    includes the counting version of the \emph{Unbounded Subset-Sum}
    problem as a special case.
\end{enumerate}

\end{abstract}



\clearpage
\tableofcontents

\section{Introduction}

\subsection{Definitions of the Main Problems}\label{sec:problem_defs}

Let $\FC \subseteq \ZZ[x_1,\dots,x_n]$ be a finite set of $m$ nonzero polynomials, and let $D := \sum_{f \in \FC} \deg f$.
Consider the system
\begin{align}
    &\begin{cases}
        f(x) \not= 0,\quad \forall f \in \FC\\
        x \in \ZZ^n.
    \end{cases}\label{avoid_sys}\tag{PolyAvoid}
\end{align}
The system \eqref{avoid_sys} has infinitely many solutions, and it is interesting to find solutions of small norm (we are mainly interested in the $\norm{\cdot}_1$-norm). This motivates the following problem, which we call the \emph{Polynomial Avoidance Problem}:
\begin{align}
    &\norm{x}_p \to \min \notag\\
    &\begin{cases}
        f(x) \not= 0,\quad \forall f \in \FC\\
        x \in \ZZ^n.
    \end{cases}\tag{$p$-PolyAvoid}\label{avoid_prob}
\end{align}

We call the affine version of this problem the \emph{Hyperplanes Avoiding Problem}, and it is of particular interest. Let $\AC \subseteq \ZZ^{n+1}$ be a set of $m$ pairs\footnote{In this case clearly $D = m$.} $(a,a_0)$ with $a \in \ZZ^n\setminus\{\BZero\}$ and $a_0 \in \ZZ$. The problem is as follows:
\begin{align}
    &\norm{x}_p \to \min \tag{$p$-HPlaneAvoid}\label{affine_avoid_prob}\\
    &\begin{cases}
        a^\top x \not= a_0,\quad \forall (a,a_0) \in \AC\\
        x \in \ZZ^n.
    \end{cases}\tag{HPlaneAvoid}\label{affine_avoid_sys}
\end{align}
In other words, we are trying to find a minimum-norm integer vector that does not lie in any of the $m$ given hyperplanes. It is also interesting to consider the \emph{homogeneous versions} of the problem \eqref{affine_avoid_prob}, when $a_0 = 0$ for every $(a,a_0) \in \AC$. In this case, we are trying to find an integer vector of the smallest norm that does not lie in any of the $m$ given $(n-1)$-dimensional subspaces.

\subsection{Motivation: Counting Integer Points in Polyhedra}\label{sec:motivation}

The problem \eqref{affine_avoid_prob}, which is the special case of the problem \eqref{avoid_prob}, naturally appears as a subproblem in algorithms for counting integer points in polyhedra.
Let us briefly explain how it arises.

\subsubsection{Barvinok's Algorithm}\label{sec:counting_barvinok}

Consider a rational polytope $\PC$ defined by a system of linear inequalities. Barvinok’s seminal work \cite{Barv_Original_Conf,Barv_Original} introduced an algorithm for computing $\PC \cap \ZZ^n$ that runs in polynomial time when the dimension is fixed. Variants of Barvinok’s algorithm can be found in \cite{OnBarvinoksAlg_Dyer,BarvPom,EffectiveCounting,BarvWoods}, while detailed accounts of the algorithm and its theoretical foundations appear in \cite{BarvBook,counting_Lasserre_book}. The algorithm represents \(\PC\cap\ZZ^n\) by a rational generating function. More precisely, Barvinok's algorithm computes a set of indices $\IC$, and for each $i \in \IC$, it computes a number $\epsilon^{(i)} \in \ZZ$ and vectors $v^{(i)}, u_1^{(i)}, \dots, u_n^{(i)} \in \ZZ^n$ such that
\begin{equation}\label{Barvinok_genf_sketch}
    \sum\limits_{x \in \PC \cap \ZZ^n} z^x = f_{\PC}(z) := \sum\limits_{i \in \IC} \epsilon^{(i)} \cdot \frac{z^{v^{(i)}}}{\bigl(1 - z^{u_1^{(i)}}\bigr) \cdot \ldots \cdot \bigl(1 - z^{u_n^{(i)}}\bigr)}.
\end{equation}
Here, the notation $z^x$ means $z^x = z_1^{x_1} \cdot \ldots \cdot z_n^{x_n}$. The right-hand side of \eqref{Barvinok_genf_sketch}, i.e. the function $f_{\PC}(z)$, is called the \emph{short rational generating function of $\PC \cap \ZZ^n$}. Since the left-hand side of \eqref{Barvinok_genf_sketch} is a finite sum, the point $z = \BUnit$ is a removable singularity of $f_{\PC}(z)$. Therefore,
\begin{equation}\label{Barvinok_lim_sketch}
    \abs{\PC \cap \ZZ^n} = \lim\limits_{z \to \BUnit} f_{\PC}(z).
\end{equation}
One way to evaluate this limit is to compute a vector $c \in \ZZ^n$ such that $c^\top u^{(i)}_j \not= 0$, for every $i \in \IC$ and $j \in \intint n$. Note that $c$ is a solution of the homogeneous system \eqref{affine_avoid_sys} with $\AC = \bigl\{ (u^{(i)}_j,0) \bigr\}$, and $m = \abs{\AC} \leq n \abs{\IC}$. Using the substitution $z_i \to e^{\tau \cdot c_i}$, the function $f_{\PC}(z)$ transforms to the function $\hat f_{\PC}(\tau)$, depending on the single complex variable $\tau$, defined by
\begin{equation}\label{exp_gen_fun_sketch}
    \hat f_{\PC}(\tau) = \sum\limits_{i \in \IC} \epsilon^{(i)} \cdot \frac{e^{\langle c, v^{(i)} \rangle \cdot \tau}}{\bigl(1 - e^{  \langle c, u_1^{(i)} \rangle \cdot \tau}\bigr) \cdot \ldots \cdot \bigl(1 - e^{ \langle c, u_n^{(i)} \rangle \cdot \tau}\bigr)}.
\end{equation}
Now, since $\hat f_{\PC}$ is analytic, the limit \eqref{Barvinok_lim_sketch} equals the $[\tau^0]$-term of the Taylor series for $\hat f_{\PC}(\tau)$:
\begin{equation*}
    \abs{\PC \cap \ZZ^n} = \lim\limits_{\tau \to 0} \hat f_{\PC}(\tau) = [\tau^0]\hat f_{\PC}.
\end{equation*}

\noindent Setting $\alpha_i = \langle c, v^{(i)} \rangle$ and $\beta_{i j} = \langle c, u^{(i)}_j \rangle$,  we obtain the following exact formula\footnote{See, for example, \cite[Chapter~14]{BarvBook}.}:
\begin{equation}\label{eq:exact_counting_todd}
    \abs{\PC \cap \ZZ^n} = \sum\limits_{i \in \IC} \frac{1}{\beta_{i 1} \cdot \ldots \cdot \beta_{i n}} \sum\limits_{j=0}^n \frac{\alpha_i^j}{j!} \cdot \toddp_{n-j}(\beta_{i 1}, \dots, \beta_{i n}),
\end{equation}
where \(\toddp_k\) denotes the \emph{homogeneous Todd polynomial} of degree \(k\).

So, for practical computation, it is highly preferable to compute the vector $c$ satisfying $c^\top u^{(i)}_j \not= 0$ with the smallest possible norm, because it will reduce the size of the numbers $\{\alpha_i,\beta_{i j}\}$, which in turn will reduce the computational cost of evaluating the Todd polynomials and the general formula \eqref{eq:exact_counting_todd}. 

We hope that our result can significantly accelerate the evaluation part of the Barvinok-type algorithms. Theoretically, it reduces the size of rational numbers during the evaluation phase from $n \log (\nu n)$ to $\log \nu$ in comparison to Barvinok's original approach. From a practical perspective, in \Cref{sec:experiments} we compare our method with random sampling on generated homogeneous instances under comparable computational budgets. Random sampling yields a strictly smaller mean $\ell_1$-norm for 25 of the 30 parameter pairs and the same mean norm for the remaining five; relative to the deterministic values, the mean reduction across parameter pairs is approximately $22.6\%$, with a median of $17.9\%$. Thus, random sampling performs better on these generated instances, whereas our deterministic method provides a rigorous worst-case guarantee on the solution norm.

\begin{remark}
    This paper does not consider dual-type algorithms for counting integer points in polyhedra. A comprehensive overview of this approach can be found in Lasserre’s book~\cite{counting_Lasserre_book}.
\end{remark}

\begin{remark}
    Related applications of short rational generating functions to integer optimization, through the computation of degrees and leading coefficients of black-box polynomials, are discussed by Zeron and De Loera \cite{BlackBoxPolynomials_ZeronDeLoera}.
    To specialize such a rational representation to one variable by a monomial substitution, its exponent vector must be nonorthogonal to every denominator direction; otherwise, a denominator factor vanishes identically. Choosing such a vector therefore leads to the same homogeneous hyperplane avoidance problem as in the counting application above.
\end{remark}

\subsubsection{An Alternative Parameterization for Counting Algorithms}\label{sec:counting_alternative}

Assume that the polyhedron $\PC$ is defined by a system $A x \leq b$, for given $A \in \ZZ^{m \times n}$ and $b \in \QQ^m$. The computational complexity of Barvinok's algorithm measured by the number of arithmetic operations can be bounded by
\begin{equation*}\label{Barvinoks_complexity}
    \nu \cdot O(\log \Delta)^{n \ln n},
\end{equation*}
where $\nu$ is the maximum size of a triangulation of the normal fan of $\PC$, and $\Delta$ is the maximum value of the rank-order subdeterminants of $A$.

However, there is an  alternative algorithmic approach to integer point counting, which allows obtaining complexity bounds of the type $\poly(\nu,n,\Delta)$. It was developed in a series of papers \cite{CountingFixedM,Counting_FPT_Delta,Counting_FPT_Delta_corrected,SparseILP_Gribanov,Parametric_Counting_Grib}.\footnote{For the latest perspective see \cite{SparseILP_Gribanov}, for the parametric case see \cite{Parametric_Counting_Grib}, the paper \cite{Counting_FPT_Delta_corrected} is a correction of \cite{Counting_FPT_Delta}.} 
In this alternative approach, 
the norm of the feasible solution of the problem \eqref{affine_avoid_prob} is a multiplicative factor in the bound on its computational complexity. More precisely, the following result was obtained in \cite{SparseILP_Gribanov}. 
\begin{theorem}[D.~Gribanov, I.~Shumilov, D.~Malyshev \& N.~Zolotykh \cite{SparseILP_Gribanov}]\label{delta_counting_complexity_prop}
    Assume that, for every collection $\AC$ of size $m$, there exists a solution $x$ of the system \eqref{affine_avoid_sys} with $\norm{x}_1 \leq L(m,n)$. Assume additionally that such an \(x\) is available at no computational cost. Then the number $\abs{\PC \cap \ZZ^n}$ can be computed using 
    \begin{equation*}
        O\left( \nu \cdot L(\nu n, n) \cdot n^2 \cdot \Delta^3 \right) \quad\text{arithmetic operations}.
    \end{equation*}
\end{theorem}
It was shown in \cite{SparseILP_Gribanov} that $L(m,n) \leq n m$, and such $x$ can be constructed by a randomized polynomial-time algorithm with $O(n m)$ operations. This implies that the counting complexity can be roughly estimated by $O\bigl(\nu^2 \cdot n^4 \cdot \Delta^3 \bigr)$. In this paper, we show that
\begin{equation*}
    L(m,n) \leq \min \left\{ \frac{m+n}{2}, m \right\},
\end{equation*}
and such $x$ can be constructed by a deterministic $O(n m)$-time algorithm. The latter yields the counting complexity 
\begin{equation}\label{eq:improved_counting_cb}
O\bigl( \nu^2 \cdot n^3 \cdot \Delta^3 \bigr) \quad\text{operations},
\end{equation}
which is the main theoretical application of our results. 

\subsubsection{Possible Applications of the Alternative Counting Approach}\label{sec:counting_apps}

The parameter $\nu$ from the formula \eqref{eq:improved_counting_cb} can be estimated in different ways, leading to various special cases and applications of the new approach. We present some of them and highlight the cases where the results of this paper yield improvements.

\paragraph{Counting Solutions to ILPs of bounded codimension}\label{sec:bounded_codim}

Assume that the polyhedron is defined by the system $A x = b,\, x \in \RR_{\geq 0}^n$, where $A \in \ZZ^{k \times n}$, $b \in \QQ^k$, and $\rank A = k$. It is natural to refer to the parameter $k$ as the \emph{codimension} of $\PC$. Similarly, for polyhedra defined by the systems $A x \leq b$ with $A \in \ZZ^{n \times d}$ and $\rank(A) = d$, we can assume that $k = n - d$. Clearly, $\nu \leq \binom{n}{k}$, which yields a counting algorithm with a computational complexity bound of  
\begin{equation*}
    O\left(\frac{n}{k}\right)^{2 k} \cdot n^3 \cdot \Delta^3 \quad\text{operations}.
\end{equation*}
This improves the previous bound in \cite{SparseILP_Gribanov} by a factor of $n$.

\paragraph{Counting solutions of the unbounded Subset-Sum problem}\label{sec:ub_subsetsum}

The preceding result implies that the solutions of the \emph{unbounded Subset-Sum problem} 
\begin{equation*}
    \begin{cases}
        w_1 x_1 + \dots + w_n x_n = w_0\\
        x \in \ZZ_{\geq 0}^n
    \end{cases}
\end{equation*}
can be counted using $O(n^5 \cdot w_{\max}^3)$ operations. Here, $w_i$
are positive integers and
$w_{\max} = \max_{i \in \intint n} w_i$. Note that the complexity
bound does not depend on $w_0$. This improves by a factor of $n$ the
$O(n^6 \cdot w_{\max}^3)$ complexity bound obtained in our previous
work \cite{SparseILP_Gribanov}.

\paragraph{ILP parameterized by the lattice determinant}\label{sec:ILP_latdet}

Let $\Delta_{\text{Lat}} = \sqrt{\det A^\top A}$. In the case where the polyhedron is defined by the system in standard form $A x = b,\, x \geq 0$, we set $\Delta_{\text{Lat}} = \sqrt{\det A A^\top}$. This parameter has received considerable attention in recent works. For instance, Aliev, Celaya et al. \cite{SparsityAndGapTransf} showed that the integrality gap is bounded by $\Delta_{\text{Lat}}$. Aliev and Henk \cite{DiagonalFrobenius} showed that the Frobenius diagonal number can also be estimated using $\Delta_{\text{Lat}}$. In Aliev, De~Loera et al. \cite{SupportIPSolutions}, the logarithm of $\Delta_{\text{Lat}}$ is used to bound the sparsity of an optimal ILP solution.
    
In turn, formula \eqref{eq:improved_counting_cb} directly yields a counting algorithm with the computational complexity bound  
\begin{equation*}
    n^3 \cdot \Delta^{O(1)}_{\text{Lat}} \quad\text{operations},
\end{equation*}
which follows from the trivial inequalities $\Delta \leq \Delta_{\text{Lat}}$ and $\nu \leq \Delta^2_{\text{Lat}}$. Obtaining a more precise polynomial dependence of computational complexity on $\Delta_{\text{Lat}}$ requires a separate study.

\paragraph{Sparse ILP}\label{sec:sparse_ILP}

Suppose that either the \(\ell_1\) norms of all rows or the \(\ell_1\) norms of all columns\footnote{It is sufficient to bound the $\ell_1$-norm for the columns of $n \times n$ nondegenerate submatrices of $A$.} of \(A\) are bounded by \(\gamma\). According to \cite{SparseILP_Gribanov}, $\nu = \gamma^{O(n)}$, which implies that the counting problem can be solved in $\gamma^{O(n)}$ operations. For instance, when the matrix $A$ has bounded entries and a bounded number of nonzeros per row or column, counting can be performed in $2^{O(n)}$ operations. In the case where only the matrix entries are bounded, the complexity of the counting problem becomes $n^{O(n)}$. See \cite{SparseILP_Gribanov} for further details.

\subsection{Our Contributions and Related Work}

We summarize our contributions as follows.
\begin{enumerate}

    \item We show that, for the \(\ell_1\)-norm, the problem \eqref{affine_avoid_prob} has a feasible solution $x$ with 
    \begin{equation}\label{eq:resulting_D_bound_affine}
        \norm{x}_1 \leq \min\left\{ \frac{m+n}{2}, m \right\}, \quad\text{where $m = \abs{\AC}$.}
    \end{equation}
    Additionally, we show that such a solution can be constructed by a deterministic polynomial-time algorithm with 
    \begin{equation*}
        O(n m) \quad \text{operations;}
    \end{equation*}
    see \Cref{affine_approx_alg_th} in \Cref{sec:approx}. The bound is tight, see \Cref{approx_alg_optimality_prop} in the discussion afterward.
    
    This is a significant improvement over the previous $O(n m)$-time randomized algorithm of \cite{SparseILP_Gribanov}, which computes $x$ with a guarantee $\norm{x}_{1} \leq n m$. In contrast, the original approach of A.~Barvinok searches for $x$ in the form $x = (1 , t, t^2, \dots, t^{n-1})$. Since, for each $a \in \AC$, $a^\top x = \sum_{i = 1}^{n} a_i \cdot t^{i-1}$ is a polynomial of degree at most $n-1$, there exists a suitable $t$ with $t \leq n m$. However, this reasoning can guarantee only $\norm{x}_1 = O\bigl((n \cdot m)^{n-1}\bigr)$;

    \item As a generalization of the above result, we show that, with respect to the $\norm{\cdot}_1$-norm, the problem \eqref{avoid_prob} has a feasible solution $x$ with 
    \begin{equation}\label{eq:resulting_D_bound}
        \norm{x}_1 \leq \min\left\{ \frac{D+n}{2}, D \right\}, \quad\text{where $D = \sum_{f \in \FC} \deg f$.}
    \end{equation}
    Additionally, we show that such a solution can be constructed by a deterministic pseudopolynomial-time algorithm\footnote{The pseudopolynomial dependence arises because the exponents of the variables in each monomial are encoded in binary. Thus, the algorithm runs in polynomial time when the degree is bounded by a constant or by a polynomial in \(n\).} whose
    \begin{equation*}
        \text{bit-complexity is}\; \poly(D, \inputsize);
    \end{equation*}
    see \Cref{approx_alg_th} in \Cref{sec:approx}. The bound is tight, since it is tight in the affine case, see \Cref{approx_alg_optimality_prop} in the discussion afterward;
    
    \item We establish an algorithmic variant of the Combinatorial
    Nullstellensatz (\Cref{nullalgo_Alon_th}) and use it to obtain our
    deterministic algorithms. This variant may be of independent interest. The classical Combinatorial Nullstellensatz guarantees the existence
    of a point at which a polynomial does not vanish, but does not in
    general provide an efficient algorithm for finding such a point;
    see Gnang~\cite{AlonNull_Hardness}. The key to our construction is
    a weaker requirement: we seek a nonvanishing point with bounded
    $\ell_1$-norm, rather than a point in a prescribed Cartesian product;

    \item We show that, for any $p \in \bigl(\RR_{\geq 1} \cup \{\infty\}\bigr)$, the problem \eqref{affine_avoid_prob} is $\NP$-hard.
    See \Cref{avoid_NPhard_th};

    \item We show that the problem of computing $\abs{\PC \cap \ZZ^n}$ for a polyhedron $\PC$ defined by the system $A x \leq b$ ($A \in \ZZ^{m \times n}$ and $b \in \QQ^m$) can be solved using $O\bigl( \nu^2 \cdot n^3 \cdot \Delta^3 \bigr)$ arithmetic operations, where $\nu$ is the maximum size of a normal fan triangulation of $\PC$, and $\Delta$ is the maximum value of rank-order subdeterminants of $A$. This improves upon the $O\bigl(\nu^2 \cdot n^4 \cdot \Delta^3\bigr)$-time algorithm of \cite{SparseILP_Gribanov}. See \Cref{sec:counting_alternative};
    
    \item As a consequence, we improve by a factor of $n$ the complexity
    bound obtained in our previous work \cite{SparseILP_Gribanov} for
    counting integer points in polyhedra of bounded codimension; see
    \Cref{sec:bounded_codim}. In particular, for the counting version of
    the \emph{Unbounded Subset-Sum} problem, we reduce the bound from
    $O(n^6 \cdot w^3_{\max})$ to $O(n^5 \cdot w^3_{\max})$ arithmetic
    operations, where $w_{\max}$ is the maximum item weight; see
    \Cref{sec:ub_subsetsum};
    
    \item As an additional straightforward application, we obtain a $\bigl(n^3 \Delta_{\text{Lat}}^{O(1)}\bigr)$-time FPT algorithm to calculate $\abs{\PC \cap \ZZ^n}$, where $\Delta_{\text{Lat}} = \sqrt{\det A^\top A}$. See Section \ref{sec:ILP_latdet};

    \item Our results can also be used to reduce the worst-case bit lengths of the rational numbers arising during the evaluation phase of Barvinok-type algorithms in the worst case; see \Cref{sec:counting_barvinok}.
\end{enumerate}

It is easy to see that the bound \eqref{eq:resulting_D_bound} is tight already in the affine case, when $D = m$; see \Cref{eq:resulting_D_bound_affine}. More precisely, the bound $\norm{x}_1 \leq m$ is tight for $m \leq n$, and the bound $\norm{x}_1 \leq (m+n)/2$ is tight for $m > n$.
\begin{proposition}\label{approx_alg_optimality_prop}
    If $m > n$, there exists a family of systems \eqref{affine_avoid_sys} such that $\norm{x}_1 \geq (m+n)/2$ for every solution $x$. Similarly, for $m \leq n$, there exists a family of systems such that $\norm{x}_1 \geq m$ for every solution $x$.
\end{proposition}
\begin{proof}
    Consider first the case $m > n$. Fix a positive integer $k$. The desired system consists of the constraints $x_i \not= j$, for every $i \in \intint n$ and $j \in \intint[-k]{k}$. So, the total number of constraints is $m = (2k + 1) \cdot n$. It is easy to see that $\abs{x_i} \geq k + 1$, for every $i \in \intint n$ and every solution $x$ of the system. Therefore, $\norm{x}_1 \geq (k+1) \cdot n = (m + n)/2$.

    Consider now the case $m \leq n$. The desired system consists of the constraints $x_i \not= 0$, for  $i \in \intint m$. Clearly, the first $m$ coordinates of any solution $x$ are nonzero, which implies that $\norm{x}_1 \geq m$.
\end{proof}

A related asymptotic covering result was obtained by Balko, Cibulka,
and Valtr~\cite[Theorem~2.3]{SubspaceCovering_Balko}. They proved that,
for a full-rank lattice $\Lambda$ and a centrally symmetric convex body
$K$, the minimum number of $k$-dimensional affine subspaces required to
cover $\Lambda \cap K$ can be expressed, up to constants depending on the
dimension, in terms of the successive minima of $\Lambda$ with respect
to $K$. In particular, for $\Lambda = \ZZ^n$, $K = R \cdot \BB_1$, and
$k = n-1$, their result shows that the minimum number of affine
hyperplanes needed to cover $\ZZ^n \cap R \cdot \BB_1$ is $\Theta_n(R)$.
Consequently, it implies an $O_n(m)$ existence bound for the general,
not necessarily homogeneous, system \eqref{affine_avoid_sys} when the dimension
is fixed. In contrast,
our result gives the explicit dimension-uniform bound
$\norm{x}_1 \leq \min \left\{(m+n)/2,m\right\}$, proves that this bound is tight, and provides a deterministic construction using $O(nm)$ arithmetic operations.

For the homogeneous version, however, the dependence of the solution norm on \(m\) can be slightly improved if no efficient construction is required. This observation is based on the following result of I.~B{\'a}r{\'a}ny, G.~Harcos, J.~Pach, \& G.~Tardos \cite{SubspaceCovering_Barany}. Let $\BB_1$ be the unit ball with respect to $\norm{x}_1$ and $g(r)$ denote the minimum number of $(n-1)$-dimensional linear subspaces needed to cover all points of the set $r \cdot \BB_1 \cap \ZZ^n$. 
\begin{theorem}[I.~B{\'a}r{\'a}ny, G.~Harcos, J.~Pach, and G.~Tardos \cite{SubspaceCovering_Barany}]\label{avoid_subspace_cover_th}
    There exist absolute constants $C_1$ and $C_2$ such that
    \begin{equation*}
        C_1 \cdot \frac{1}{n^2} \cdot r^{\frac{n}{n-1}} \leq g(r) \leq C_2 \cdot 2^n \cdot r^{\frac{n}{n-1}}.
    \end{equation*}
\end{theorem}
Note that the original work \cite{SubspaceCovering_Barany} contains a  more general result concerning arbitrary convex bodies in $\RR^n$, albeit  with a  worse dependence on $n$. \Cref{avoid_subspace_cover_th} is a straightforward adaptation of the original proof to the case of $\BB_1$. It follows then that the system \eqref{affine_avoid_sys} always has a solution whose norm has slightly better dependence on \(m\), but worse dependence on \(n\).
\begin{corollary}\label{avoid_exist_th}
    The homogeneous system \eqref{affine_avoid_sys} has a solution $x$ such that
    $$
        \norm{x}_1 = O \bigl( n^2 \cdot m^{\frac{n-1}{n}} \bigr).
    $$
\end{corollary}

At the same time, the theorem implies that solutions of significantly smaller norm do not exist in general. In particular, it implies that our constructive bound $\norm{x}_1 \leq \min \left\{(m+n)/2,m\right\}$ is almost optimal with respect to $m$ even in the homogeneous case.
\begin{corollary}\label{avoid_nonexist_th}
    There exists a homogeneous system \eqref{affine_avoid_sys} such that, for every solution $x$, 
    $$
        \norm{x}_1 = \Omega\bigl( \frac{1}{2^n} \cdot m^{\frac{n-1}{n}} \bigr).
    $$
\end{corollary}

The corresponding covering problem has also been studied for linear
subspaces of arbitrary dimension. Balko, Cibulka, and
Valtr~\cite{SubspaceCovering_Balko} determined the corresponding covering exponent up to a
subpolynomial factor. Sudakov and
Tomon~\cite{EvasiveSets_SudakovTomon} subsequently eliminated the subpolynomial factor and proved that, for fixed $n$ and $1 \leq k < n$, the minimum number of
$k$-dimensional linear subspaces needed to cover the grid
$(\intint N)^n$ is
\begin{equation*}
    \Theta_{n,k}\left(N^{\frac{n(n-k)}{n-1}}\right).
\end{equation*}
For $k = n-1$, this bound has the same radius exponent $n/(n-1)$ as
\Cref{avoid_subspace_cover_th}. After converting between the
$\ell_\infty$ and $\ell_1$ norms in fixed dimension, it yields the same
$m^{(n-1)/n}$ dependence for the homogeneous version of our avoidance
problem. These are asymptotic covering-number results and do not address
the efficient construction of an avoiding point for a given family of
hyperplanes.

We need to mention another deep generalization of the problem \eqref{affine_avoid_prob} in its homogeneous form appearing in the literature.
Fukshansky~\cite{HeightAvoiding} studied a substantially more general homogeneous avoidance problem. His main result gives an upper bound on \(\norm{x}_\infty\) for a vector \(x\) in an ambient lattice \(\Omega\subseteq\ZZ^n\) that avoids a prescribed collection of sublattices $\Lambda_1, \dots, \Lambda_m \subseteq \Omega$. This inequality connects the heights of bases of the input lattices $\Omega, \Lambda_1, \dots, \Lambda_m$ such that, if the heights of $\Lambda_1, \dots, \Lambda_m$ are sufficiently large, then $\norm{x}_{\infty} = O_n(\sqrt{m})$. We will not provide the exact statement by \cite{HeightAvoiding} in our work, and only provide its straightforward corollary with respect to our problem \eqref{affine_avoid_prob}.
\begin{theorem}[\cite{HeightAvoiding}]
    The homogeneous system \eqref{affine_avoid_sys} has a solution $x$ such that
    \begin{equation*}
        \norm{x}_{\infty} \leq n + \left(\frac{3}{2}\right)^{n-1} \cdot  n^n \cdot \left(\sqrt{m} + \sum_{a \in \AC} \frac{1}{\norm{a}_{\infty}}\right) = O(n)^n \cdot \left(\sqrt{m} + \sum_{a \in \AC} \frac{1}{\norm{a}_{\infty}}\right). 
    \end{equation*}
\end{theorem}
One can see that if the elements of the collection $\AC$ have sufficiently large $\norm{\cdot}_{\infty}$-norm, then the dependence on $m$ becomes $O_n(\sqrt{m})$ which outperforms both our main result and \Cref{avoid_exist_th} of \Cref{avoid_subspace_cover_th}. However, this comes at the cost of $O(n)^n$-dependence in $n$.

\section{An Approximate Solution via the Combinatorial Nullstellensatz}\label{sec:approx}

The existence of a small-norm solution $x$ of \eqref{avoid_sys} can be guaranteed by the  \emph{Combinatorial Nullstellensatz} due to N.~Alon \cite{AlonNull}. 
\begin{theorem}[N.~Alon \cite{AlonNull}]\label{null_Alon_th}
    Let $f(x) \in \FF[x_1, \dots, x_n]$ be a nonzero polynomial with coefficients in the field $\FF$, and let $x_1^{d_1}  \ldots x_n^{d_n}$ be a monomial of $f$ of the maximum total degree.  For $i \in \{1,2,\dots,n\}$, let $S_i \subseteq \FF$ be a set  with $\abs{S_i} \geq d_i + 1$. Then, there exists $y \in S_1 \times \ldots \times S_n$ such that $f(y) \not= 0$. 
\end{theorem}

\begin{proposition}\label{small_norm_via_null_prop}
    The system \eqref{avoid_sys} has a solution $y$ such that 
    \begin{equation*}
        \norm{y}_1 \leq \min\left\{ \frac{D+n}{2}, D \right\}.    
    \end{equation*}
\end{proposition}
\begin{proof}
    Consider a polynomial
    \begin{equation}\label{polynew}
        F(x) = \prod\limits_{f \in \FC} f.
    \end{equation} 
    
    \noindent Let $x_1^{d_1} \cdot \ldots \cdot x_n^{d_n}$ be its monomial of the highest total degree. Clearly, $x$ is a solution to \eqref{avoid_sys} if and only if $F(x) \not= 0$. Define the sets
    \begin{equation*}\label{deg_sets_S}
        S_i = \bigl\{-\lceil d_i/2 \rceil, \dots, -1, 0, 1, \dots, \lceil d_i/2 \rceil \bigr\},\quad\text{for $i \in \intint n$.}
    \end{equation*}
    Note that $\abs{S_i} \geq d_i + 1$, and, for any $y \in S_1 \times \ldots \times S_n$,
    \begin{equation*}
        \norm{y}_1 \leq \lceil d_1/2 \rceil + \ldots + \lceil d_n/2 \rceil \leq \min\left\{ \frac{D+n}{2}, D \right\}. 
    \end{equation*}
    Here we have used the inequalities $\lceil d_i/2 \rceil \leq (d_i + 1)/2$ and $\lceil d_i/2 \rceil \leq d_i$. Thus \Cref{null_Alon_th} implies the proposition.
\end{proof}

Although the Combinatorial Nullstellensatz guarantees the existence of such a vector \(y\), it does not in general provide an efficient algorithm for finding one; see, e.g., Gnang~\cite{AlonNull_Hardness}. We therefore introduce a weaker variant that gives the same bound on the \(\ell_1\)-norm and permits an efficient construction whenever the required oracle calls can be implemented efficiently. 

\begin{theorem}\label{nullalgo_Alon_th}
    Let \(F\in\ZZ[x_1,\ldots,x_n]\) be nonzero, and let \(D=\deg F\). Suppose that, for every \(k\in\{1,\ldots,n\}\) and every tuple \((y_1,\ldots,y_k)\in\ZZ^k\), an oracle determines whether \(F(y_1,\ldots,y_k,x_{k+1},\ldots,x_n)\) is identically zero.
    
    Then there is an algorithm that computes a point \(y\in\ZZ^n\) satisfying
    \begin{equation*}
        \norm{y}_1 \leq \min\left\{ \frac{D+n}{2}, D \right\} \quad\text{and}\quad F(y) \not= 0.
    \end{equation*}
    The algorithm uses at most $D + 2n$ calls to the oracle. Additionally, the solution $y$ generated by the algorithm is \emph{lexicographically minimal in absolute value} in the following sense: $\abs{y_1}$ is minimized over all feasible points, $\abs{y_2}$ is minimized over all feasible points with $\abs{x_1} = \abs{y_1}$, and so on.
\end{theorem}
\begin{proof}
    The algorithm consists of $n$ elementary steps, each of which eliminates one of the variables $x_i$. The first step eliminates $x_1$, the second step eliminates $x_2$, and so on. Consider the first step. Let $S = \intint[-\lceil D/2 \rceil]{\lceil D/2 \rceil}$. Examine the values \(y_1\in S\) in nondecreasing order of \(\lvert y_1\rvert\). Choose the first \(y_1\) for which \(F(y_1,x_2,\ldots,x_n)\) is not identically zero. Every rejected value \(t\in S\) yields a distinct linear factor \((x_1-t)\) of \(F\). Since $\abs{S} > D$ and the degree of $F$ is $D$, such $y_1$ exists. 
    When such $y_1$ is found, we make a substitution $\widehat F(x_2, \dots, x_n) = F(y_1, x_2, \dots, x_n)$. We claim that $\deg \widehat F \leq D - (2\abs{y_1} -1)$ for nonzero $y_1$. 
    Indeed, if $\abs{y_1} \geq 1$, then $F$ is divisible by 
    \begin{equation*}
        \prod\limits_{i = -\bigl(\abs{y_1}-1\bigr)}^{\abs{y_1}-1} (x_1 - i),
    \end{equation*}
    which proves the claim.
    
    The second step eliminates the variable $x_2$ in $\widehat F$ using the same scheme. Continuing in this way, we find $y = (y_1, \dots, y_n)^\top \in S^n$ with $F(y_1, \dots, y_n) \not= 0$, and
    \begin{equation*}
        \sum\limits_{i \colon y_i \not= 0} (2 \abs{y_i} - 1) \leq D.
    \end{equation*}
    Rearranging the preceding inequality yields $\norm{y}_1 \leq (D+n)/2$. The inequality $\norm{y}_1 \leq D$ follows from the inequalities $2 \abs{y_i} -1 \geq \abs{y_i}$, which hold for nonzero $y_i$. In turn, the total number of oracle calls is bounded by $\sum\limits_{i = 1}^n (2 \abs{y_i} + 1) \leq D + 2n$.
\end{proof}

We now apply \(\Cref{nullalgo_Alon_th}\) to \(\eqref{avoid_prob}\). It remains only to implement the corresponding oracle calls.
\begin{theorem}\label{approx_alg_th}
    Assume that every polynomial in \eqref{avoid_sys} is given in sparse form.
    Let $L$ be the total number of monomials in these polynomials, and
    let $\alpha$ be the maximum absolute value of their coefficients.
    There is a deterministic pseudo\-polynomial-time algorithm that
    computes a solution $y\in\ZZ^n$ to \eqref{avoid_sys} with the
    following properties:
    \begin{enumerate}
        \item The algorithm performs
        \begin{equation*}
            O\bigl((D + n) \cdot n L \cdot \log L \cdot \log D\bigr)
        \end{equation*}
        arithmetic operations on integers of the bit-encoding length
        \begin{equation*}
            O(\log \alpha + D \log D + \log L);
        \end{equation*}

        \item The output vector $y$ satisfies
        \begin{equation*}
            \norm{y}_1 \leq \min\left\{ \frac{D+n}{2}, D \right\};
        \end{equation*}

        \item The solution $y$ is \emph{lexicographically minimal in absolute value} in the sense described in \Cref{nullalgo_Alon_th}.
    \end{enumerate}
\end{theorem}
\begin{proof}
    Consider the polynomial $F$ defined in \eqref{polynew}. Since every $f\in\FC$ is nonzero and $\ZZ[x_1,\ldots,x_n]$ is an integral domain, we have $F\not\equiv0$ and $\deg F=D$. Moreover,
    $F(y)\neq0$ if and only if $f(y)\neq0$ for every $f\in\FC$.
    Thus, it remains to implement the oracle from \Cref{nullalgo_Alon_th} without expanding $F$.

    Recall the algorithm from \(\Cref{nullalgo_Alon_th}\), and suppose that the \(k\)-th step has been started. Suppose that the algorithm has chosen \(y_1,\ldots,y_k\) such that
    \begin{gather*}
        \widehat F(x_{k+1}, \dots, x_n) := \prod_{f\in\FC} \hat f(x_{k+1}, \dots, x_n) \not\equiv 0, \quad\text{where}\\
        \hat f(x_{k+1}, \dots, x_n) = f(y_1, \dots, y_k, x_{k+1}, \dots, x_n),
    \end{gather*} 
    and each $\hat f$ is represented in sparse form.

    At step \(k+1\), the algorithm examines the elements of \(S = \intint[-\lceil D/2 \rceil]{\lceil D/2 \rceil}\) in nondecreasing order of absolute value until $\widehat F(y_{k+1}, x_{k+2}, \dots, x_n) \not\equiv 0$, which is equivalent to $\hat f(y_{k+1}, x_{k+2}, \dots, x_n) \not\equiv 0$, for every $f \in \FC$.
    
    The condition $\hat f(y_{k+1}, x_{k+2}, \dots, x_n) \not\equiv 0$ can be checked by substituting \(x_{k+1}\gets y_{k+1}\), combining like monomials, and testing whether all resulting coefficients vanish. Since the number of monomials in each polynomial cannot increase during the algorithm, the binary search tree used to combine like monomials contains at most $L$ nodes. Thus, inserting or locating a monomial requires $O(\log L)$ comparisons. Each monomial is represented by an integer vector with at most $n$ coordinates, each bounded by $D$, so comparing two such vectors requires $O(n \log D)$ operations. Consequently, processing all the polynomials $\hat f$ requires $O(n L \cdot \log L \cdot \log D)$ operations per oracle call. Computing \(y_{k+1}^s\), where \(s\leq D\), by repeated squaring requires an additional $O(L \log D)$ operations, which is absorbed into the preceding bound. By \Cref{nullalgo_Alon_th}, the algorithm uses at most $D + 2n$ oracle calls. Hence, its total arithmetic complexity is
    \begin{equation*}
        O\bigl((D + n) \cdot n L \cdot \log L \cdot \log D\bigr).
    \end{equation*}
    
    Since every queried coordinate has absolute value at most \(\lceil D/2\rceil\) and every monomial has total degree at most \(D\), each contribution obtained after substitution has bit length \(O(\log \alpha + D\log D)\). Combining at most \(L\) such contributions increases this bound by at most \(O(\log L)\). Therefore, the bit-encoding length of the integers involved in the computation can be bounded by
    \begin{equation*}
        O(\log \alpha + D \log D + \log L).
    \end{equation*}

    We have therefore implemented every oracle call within the stated complexity bounds. Applying \Cref{nullalgo_Alon_th} gives the claimed bound on $\norm{y}_1$ and the lexicographic minimality of the output. The proof is finished.

\end{proof}

We now apply \Cref{nullalgo_Alon_th} to the polynomial \eqref{polynew}. In that case, the whole procedure can be made very efficient.
\begin{theorem}\label{affine_approx_alg_th}
    There exists a deterministic polynomial-time algorithm which computes a solution $y$ of the system \eqref{affine_avoid_sys} such that the following claims are satisfied:
    \begin{enumerate}
        \item The algorithm makes $O(n  m)$ arithmetic operations with integer numbers of bit-encoding size bounded by $O(\log \alpha + \log m + \log n)$, where $\alpha$ is the maximum absolute value of the coefficients;
        \item \begin{equation*}
            \norm{y}_1 \leq \min \left\{ \frac{m+n}{2}, m \right\};
        \end{equation*}
        \item The solution $y$ is \emph{lexicographically minimal in absolute value} in the sense described in \Cref{nullalgo_Alon_th}.
    \end{enumerate}
\end{theorem}
\begin{proof}
Recall
\begin{equation*}
   F(x) = \prod\limits_{(a,a_0) \in \AC} \bigl(a^\top x - a_0\bigr).
\end{equation*}
Again, $x$ is a solution \eqref{affine_avoid_sys} if and only if $F(x) \not= 0$, and $\deg(F) = m$. By \Cref{nullalgo_Alon_th}, the only thing that we need to check is that all oracle calls during a single step of the algorithm shall cost us $O(m)$ operations. Indeed, at each step we shall maintain a list of linear factors, and change each factor by substituting the values of $x_1,$ $x_2,$ etc. Checking if each factor is identically $0$ after a given substitution requires $O(1)$ operations, and the polynomial $F$ is not identically $0$ if and only if each of the factors is not identically $0$.

The detailed description is given as follows. For any $k \in \intint n$, at the beginning of the $k$-th step, we maintain the polynomial $\widehat F \in \ZZ[x_k, \dots, x_n]$ represented by a finite set of integer vectors $\widehat\AC \subseteq \ZZ^{n-k+2}$ such that
\begin{equation*}
    \widehat F(x) = \prod\limits_{(a,a_0) \in \widehat\AC} (a^\top x - a_0).
\end{equation*}
For each vector $(a,a_0) \in \widehat\AC$, we store a linked list $\LC_{(a,a_0)}$ of nonzero entries of $a$. More precisely, $\LC_{(a,a_0)}$ consists of the pairs $(i, a_i)$ for each nonzero $a_i$, $i \in \intint[k]{n}$. The pairs inside $\LC_{(a,a_0)}$ are stored in increasing order of the first element of a pair. Initially, at the beginning of the first step, we just assign $\widehat F \gets F$ and $\widehat\AC \gets \AC$. Clearly, for every $(a,a_0) \in \widehat\AC$, the lists $\LC_{(a,a_0)}$ can be initialized with $O(n m)$ operations.

We maintain the following invariant. For each $k \in \intint n$, after the $k$-th step has been completed, the following conditions have to be satisfied:
\begin{enumerate}
    \item For some integer constant $C$, \begin{equation*}
        F(y_1, \dots, y_k, x_{k+1}, \dots, x_n) = C \cdot \widehat F(x_{k+1}, \dots, x_n) \not\equiv 0;
    \end{equation*}

    \item For each $(a,a_0) \in \widehat\AC$, the list $\LC_{(a,a_0)}$ is nonempty.
\end{enumerate}

Now, let us describe how the $k$-th step can be performed. The step is aimed to find $y_k \in S := \left\{-\lceil m/2 \rceil, \dots, \lceil m/2 \rceil\right\}$ such that $\widehat F(y_k, x_{k+1}, \dots, x_n) \not\equiv 0$, and $\abs{y_k}$ is minimized. To this end, we construct an indicator vector $v \in \{0,1\}^{S}$, whose elements are indexed by the elements of $S$, such that $v_i = 1$ if and only if $\widehat F(i, x_{k+1}, \dots, x_n) \equiv 0$. Assuming that $v$ is constructed, we can just set $y_k \gets i$ for $i \in S$ with $v_i = 0$ and having a minimum absolute value. Due to the reasoning from the first part of the proof, such an $i$ always exists. Note that such $i$ can be computed with $O(m)$ operations. We now describe how to construct the vector $v$. To this end, we search for linear factors of $\widehat F$ of the form $a_k \cdot x_k - a_0$. If such a factor is found, and if $a_0/a_k \in S$, then we set $v_{i} \gets 1$, where $i = a_0/a_k$. To perform this search, we just need to scan all elements $(a,a_0) \in \widehat\AC$ with $\LC_{(a,a_0)}$ consisting of a single element $(k,a_k)$. Clearly, the latter can be done with $O\bigl(|\widehat\AC|\bigr) = O(m)$ operations.

Before finishing the $k$-th step, we need to maintain the invariant. In other words, we need to perform the substitution of $x_k \gets y_k$ inside $\widehat F$. To this end, for each $(a,a_0) \in \widehat\AC$, if the first element of $\LC_{(a,a_0)}$ is $(k, a_k)$, we set $a_0 \gets a_0 - a_k \cdot y_k$ and remove $(k, a_k)$ from the list $\LC_{(a,a_0)}$. If $\LC_{(a,a_0)}$ becomes empty, we remove $\LC_{(a,a_0)}$ together with the element $(a,a_0)$ of $\widehat\AC$. Clearly, this work can be done with $O\bigl(|\widehat\AC|\bigr) = O(m)$ operations. Therefore, the algorithm consists of $O(n m)$-time preprocessing, and $n$ steps consisting of $O(m)$ operations, which gives the desired complexity bound.
\end{proof}

\section{Complexity of Computing an Exact Solution}\label{sec:exact}

\begin{lemma}\label{avoid_inf_NPhard_lm}
    For $p = \infty$, the homogeneous form of the problem \eqref{affine_avoid_prob} is $\NP$-hard.
\end{lemma}
\begin{proof}
     Let us reduce the classical $\NP$-hard \emph{Chromatic Number} problem to our problem. Consider an arbitrary simple graph $G = (V, E)$ on $n$ vertices with $m$ edges. It is not hard to see that the Chromatic Number problem with respect to $G$ can be formulated as follows:
    \begin{align}
        &\norm{x}_\infty \to \min\notag\\
        &\begin{cases}
            \forall (u,v) \in E,\quad \abs{x_u} \not= \abs{x_v}\\
            x \in \mathbb{Z}^{V}.
        \end{cases}\label{chromnum_prob}
    \end{align}
    Here, for an optimal solution $x^*$ of the problem above, the chromatic number of $G$ is exactly $\norm{x^*}_{\infty} + 1$, and the colors are $\bigl\{0,1, \dots, \norm{x^*}_{\infty} \bigr\}$. A proper coloring  of the vertices is given by the vector $\abs{x^*}$. 
    
    The formulation \eqref{chromnum_prob} can be rewritten as
    \begin{align*}
        &\norm{x}_\infty \to \min\\
        &\begin{cases}
            \forall (u,v) \in E,\quad x_u - x_v \not= 0\\
            \forall (u,v) \in E,\quad x_u + x_v \not= 0\\
            x \in \mathbb{Z}^{V},
        \end{cases}
    \end{align*}
    This system is an instance of  \eqref{affine_avoid_prob} for $p = \infty.$ This proves $\NP$-hardness of the latter. 
\end{proof}

\begin{lemma}\label{avoid_finite_NPhard_lm}
    For any $p \in \RR_{\geq 1}$, the homogeneous form of the problem \eqref{affine_avoid_prob} is $\NP$-hard.
\end{lemma}
\begin{proof}
     We reduce the classical $\NP$-hard \emph{Minimum Vertex Cover} problem to our problem. Consider an arbitrary simple graph $G = (V, E)$ on \(n\) vertices with \(m\) edges. We claim that the Minimum Vertex Cover problem for $G$ can be formulated as follows:
    \begin{align*}
        &\norm{x}_p^p \to \min\\
        &\begin{cases}
            \forall (u,v) \in E,\quad x_u + x_v\ne 0\\
            x \in \mathbb{Z}^{V}.
        \end{cases}
    \end{align*}
    Indeed, given an optimal solution  $x^*$ to this system,  construct a new vector $y^*$ such that $y^*_i = 0$ if and only if $x^*_i=0$ and $y^*_i = 1$ otherwise. First, note that $\norm{x^*}_p^p\ge \norm{y^*}_p^p$ and, second, that $y^*$ is also a solution to the system above. Indeed, if $x^*_u + x^*_v\ne 0$, then at least one of $x_u,x_v$ is not $0$, and thus $y^*_u + y^*_v\ge 1$. This implies that there is an optimal solution $x^*$ to the system above with $x^*_v\in \{0,1\}$ for every $v \in V$. Finally, for $\{0,1\}$-solutions it is easy to see that any such solution corresponds to a vertex cover of $G$ and that, moreover, minimizing the $\ell_p$-norm is the same as minimizing the size of the cover. Since the displayed problem is an instance of \eqref{affine_avoid_prob}, it proves $\NP$-hardness of the latter.
\end{proof}

\Cref{avoid_inf_NPhard_lm} and \Cref{avoid_finite_NPhard_lm} yield
\begin{theorem}\label{avoid_NPhard_th}
    For any $p \in \bigl(\RR_{\geq 1} \cup \{\infty\}\bigr)$, the problem \eqref{affine_avoid_prob} is $\NP$-hard.
\end{theorem}

\section{Experimental Evaluation}\label{sec:experiments}

In this section, we provide an experimental comparison between the algorithm from \Cref{affine_approx_alg_th} and a random sampling strategy for integer points in the crosspolytope $r \mathbb{B}_1$, with $r = \left\lceil (m+n)/2 \right\rceil$. The experiments are performed on homogeneous systems \eqref{affine_avoid_sys} with integer coefficients drawn from a fixed interval. The sampling scheme works as follows. We fix a number $L$ of sampling iterations for each radius and an initial $\ell_1$-norm radius. The radius is adjusted using a bisection-type strategy. Whenever a sampled point lies on none of the hyperplanes, we halve the radius. Conversely, if after $L$ trials no point avoiding all hyperplanes is found, we set the new radius to the midpoint between the current radius and the previous one. Subsequent radius updates are performed analogously. The trial budget $L$ is calibrated so that the overall runtime of the sampling procedure is comparable to that of our algorithm from \Cref{affine_approx_alg_th}, thereby providing comparable computational budgets. As \Cref{tab:hyperplanes-time} shows, the mean running times are comparable. Regarding the $\ell_1$-norms of the computed solutions, \Cref{tab:hyperplanes-l1} shows that random sampling performs at least as well as the deterministic algorithm for every parameter pair. It yields a strictly smaller mean norm for 25 of the 30 parameter pairs and the same mean norm for the remaining five. Relative to the deterministic values, the mean reduction across parameter pairs is approximately $22.6\%$, with a median of $17.9\%$. The largest reduction is approximately $66.7\%$, attained for $(n,m)=(150,500)$, where the mean norm decreases from $3.18$ to $1.06$. Hence, our algorithm is best viewed as a deterministic method with a worst-case guarantee. This property can be advantageous in practice, because real-world instances are not necessarily governed by a uniform input distribution.

\begin{table}[ht]
\centering
\small
\setlength{\tabcolsep}{4pt}
\resizebox{\textwidth}{!}{
\begin{tabular}{rcccccccccccc}
\toprule
$m$ & \multicolumn{2}{c}{$10$} & \multicolumn{2}{c}{$50$} & \multicolumn{2}{c}{$100$} & \multicolumn{2}{c}{$150$} & \multicolumn{2}{c}{$500$} & \multicolumn{2}{c}{$1000$} \\
\cmidrule(lr){2-3}\cmidrule(lr){4-5}\cmidrule(lr){6-7}\cmidrule(lr){8-9}\cmidrule(lr){10-11}\cmidrule(lr){12-13}
$n$ & D & R & D & R & D & R & D & R & D & R & D & R \\
\midrule
$10$ & 1.00 & 1.00 & 1.07 & 1.00 & 1.13 & 1.00 & 1.35 & 1.02 & 3.05 & 1.54 & 4.74 & 2.85 \\
$50$ & 1.00 & 1.00 & 1.04 & 1.00 & 1.14 & 1.00 & 1.34 & 1.00 & 3.04 & 1.21 & 4.66 & 3.52 \\
$100$ & 1.00 & 1.00 & 1.05 & 1.00 & 1.15 & 1.00 & 1.47 & 1.00 & 2.97 & 1.00 & 4.97 & 3.30 \\
$150$ & 1.00 & 1.00 & 1.02 & 1.00 & 1.12 & 1.00 & 1.39 & 1.00 & 3.18 & 1.06 & 4.86 & 3.25 \\
$500$ & 1.00 & 1.00 & 1.05 & 1.00 & 1.11 & 1.00 & 1.37 & 1.00 & 2.92 & 1.03 & 4.72 & 3.65 \\
\bottomrule
\end{tabular}%
}
\caption{Mean \(\ell_1\) norm of the vectors returned by the two methods across the generated homogeneous instances for each parameter pair \((n,m)\). Here, D denotes the deterministic algorithm of \Cref{affine_approx_alg_th}, and R denotes the randomized sampling method.}
\label{tab:hyperplanes-l1}
\end{table}

\begin{table}[ht]
\centering
\small
\setlength{\tabcolsep}{4pt}
\resizebox{\textwidth}{!}{%
\begin{tabular}{rcccccccccccc}
\toprule
$m$ & \multicolumn{2}{c}{$10$} & \multicolumn{2}{c}{$50$} & \multicolumn{2}{c}{$100$} & \multicolumn{2}{c}{$150$} & \multicolumn{2}{c}{$500$} & \multicolumn{2}{c}{$1000$} \\
\cmidrule(lr){2-3}\cmidrule(lr){4-5}\cmidrule(lr){6-7}\cmidrule(lr){8-9}\cmidrule(lr){10-11}\cmidrule(lr){12-13}
$n$ & D & R & D & R & D & R & D & R & D & R & D & R \\
\midrule
$10$ & 0.035 & 0.068 & 0.116 & 0.121 & 0.223 & 0.228 & 0.321 & 0.326 & 1.012 & 1.036 & 2.170 & 2.190 \\
$50$ & 0.084 & 0.179 & 0.375 & 0.391 & 0.676 & 0.687 & 0.957 & 0.974 & 3.190 & 3.333 & 6.558 & 6.932 \\
$100$ & 0.144 & 0.306 & 0.585 & 0.609 & 1.380 & 1.424 & 1.917 & 1.949 & 6.102 & 6.273 & 12.805 & 14.068 \\
$150$ & 0.247 & 0.420 & 1.030 & 1.052 & 2.051 & 2.081 & 3.032 & 3.088 & 10.208 & 10.361 & 20.070 & 21.300 \\
$500$ & 0.568 & 0.989 & 2.699 & 2.753 & 5.971 & 6.108 & 8.204 & 8.389 & 26.160 & 27.341 & 49.113 & 60.036 \\
\bottomrule
\end{tabular}%
}
\caption{Mean running time (ms) of the two methods across the generated homogeneous instances for each parameter pair \((n,m)\). Here, D denotes the deterministic algorithm of \Cref{affine_approx_alg_th}, and R denotes the randomized sampling method.}
\label{tab:hyperplanes-time}
\end{table}
\section*{Statements and Declarations}

\subsection*{Competing interests}
The authors declare that they have no competing interests.

\subsection*{Data availability}
The numerical results of the experiments are reported in
Tables~\ref{tab:hyperplanes-l1} and~\ref{tab:hyperplanes-time}.
No separate dataset is available.

\subsection*{Use of artificial intelligence}
The proofs were developed without the use of AI theorem-proving tools.
AI tools were used to assist with surveying the literature and checking
the correctness of the proofs.

\bibliographystyle{plainurl}
\bibliography{grib_biblio}

\begin{thebibliography}{10}

\bibitem{SupportIPSolutions}
I.~Aliev, J.~A. De~Loera, F.~Eisenbrand, T.~Oertel, and R.~Weismantel.
\newblock The support of integer optimal solutions.
\newblock {\em SIAM Journal on Optimization}, 28(3):2152--2157, 2018.
\newblock \href {https://doi.org/10.1137/17M1162792}
  {\path{doi:10.1137/17M1162792}}.

\bibitem{SparsityAndGapTransf}
Iskander Aliev, Marcel Celaya, and Martin Henk.
\newblock Sparsity and integrality gap transference bounds for integer
  programs.
\newblock In J.~Vygen and J.~Byrka, editors, {\em Integer Programming and
  Combinatorial Optimization}, volume 14679 of {\em Lecture Notes in Computer
  Science}, pages 1--13, Cham, 2024. Springer.
\newblock \href {https://doi.org/10.1007/978-3-031-59835-7_1}
  {\path{doi:10.1007/978-3-031-59835-7_1}}.

\bibitem{DiagonalFrobenius}
Iskander Aliev and Martin Henk.
\newblock Feasibility of integer knapsacks.
\newblock {\em SIAM Journal on Optimization}, 20(6):2978--2993, 2010.

\bibitem{AlonNull}
Noga Alon.
\newblock Combinatorial {Nullstellensatz}.
\newblock {\em Combinatorics, Probability and Computing}, 8(1-2):7--29, 1999.

\bibitem{SubspaceCovering_Balko}
Martin Balko, Josef Cibulka, and Pavel Valtr.
\newblock Covering lattice points by subspaces and counting point--hyperplane
  incidences.
\newblock {\em Discrete \& Computational Geometry}, 61(2):325--354, 2019.
\newblock \href {https://doi.org/10.1007/s00454-018-9970-7}
  {\path{doi:10.1007/s00454-018-9970-7}}.

\bibitem{SubspaceCovering_Barany}
Imre B{\'a}r{\'a}ny, Gergely Harcos, J{\'a}nos Pach, and G{\'a}bor Tardos.
\newblock Covering lattice points by subspaces.
\newblock {\em Periodica Mathematica Hungarica}, 43(1--2):93--103, 2002.
\newblock \href {https://doi.org/10.1023/A:1015233631926}
  {\path{doi:10.1023/A:1015233631926}}.

\bibitem{Barv_Original_Conf}
A.~I. Barvinok.
\newblock A polynomial time algorithm for counting integral points in polyhedra
  when the dimension is fixed.
\newblock In {\em Proceedings of 1993 IEEE 34th Annual Foundations of Computer
  Science}, pages 566--572, 1993.
\newblock \href {https://doi.org/10.1109/SFCS.1993.366830}
  {\path{doi:10.1109/SFCS.1993.366830}}.

\bibitem{BarvBook}
Alexander Barvinok.
\newblock {\em Integer Points in Polyhedra}.
\newblock Zurich Lectures in Advanced Mathematics. European Mathematical
  Society, Z{\"u}rich, 2008.
\newblock \href {https://doi.org/10.4171/052} {\path{doi:10.4171/052}}.

\bibitem{BarvPom}
Alexander Barvinok and James~E. Pommersheim.
\newblock An algorithmic theory of lattice points in polyhedra.
\newblock In Louis~J. Billera, Anders Bj{\"o}rner, Curtis Greene, Rodica
  Simion, and Richard~P. Stanley, editors, {\em New Perspectives in Algebraic
  Combinatorics}, volume~38 of {\em Mathematical Sciences Research Institute
  Publications}, pages 91--147. Cambridge University Press, Cambridge, 1999.

\bibitem{BarvWoods}
Alexander Barvinok and Kevin Woods.
\newblock Short rational generating functions for lattice point problems.
\newblock {\em Journal of the American Mathematical Society}, 16(4):957--979,
  2003.
\newblock URL: \url{http://www.jstor.org/stable/30041461}.

\bibitem{Barv_Original}
Alexander~I. Barvinok.
\newblock A polynomial time algorithm for counting integral points in polyhedra
  when the dimension is fixed.
\newblock {\em Mathematics of Operations Research}, 19(4):769--779, 1994.
\newblock \href {https://doi.org/10.1287/moor.19.4.769}
  {\path{doi:10.1287/moor.19.4.769}}.

\bibitem{EffectiveCounting}
Jes{\'u}s~A. {De Loera}, Raymond Hemmecke, Jeremiah Tauzer, and Ruriko Yoshida.
\newblock Effective lattice point counting in rational convex polytopes.
\newblock {\em Journal of Symbolic Computation}, 38(4):1273--1302, 2004.
\newblock Symbolic Computation in Algebra and Geometry.
\newblock URL:
  \url{https://www.sciencedirect.com/science/article/pii/S0747717104000422},
  \href {https://doi.org/10.1016/j.jsc.2003.04.003}
  {\path{doi:10.1016/j.jsc.2003.04.003}}.

\bibitem{OnBarvinoksAlg_Dyer}
Martin Dyer and Ravi Kannan.
\newblock On {B}arvinok's algorithm for counting lattice points in fixed
  dimension.
\newblock {\em Mathematics of Operations Research}, 22(3):545--549, 1997.
\newblock \href {https://doi.org/10.1287/moor.22.3.545}
  {\path{doi:10.1287/moor.22.3.545}}.

\bibitem{HeightAvoiding}
Lenny Fukshansky.
\newblock Integral points of small height outside of a hypersurface.
\newblock {\em Monatshefte f{\"u}r Mathematik}, 147(1):25--41, 2006.
\newblock \href {https://doi.org/10.1007/s00605-005-0333-0}
  {\path{doi:10.1007/s00605-005-0333-0}}.

\bibitem{AlonNull_Hardness}
Edinah~K. Gnang.
\newblock {\em Computational aspects of the combinatorial {Nullstellensatz}
  method via a polynomial approach to matrix and hypermatrix algebra}.
\newblock PhD thesis, Rutgers, The State University of New Jersey, New
  Brunswick, NJ, 2013.

\bibitem{Counting_FPT_Delta}
D.~V. Gribanov and D.~S. Malyshev.
\newblock A faster algorithm for counting the integer points number in
  {$\Delta$}-modular polyhedra.
\newblock {\em Siberian Electronic Mathematical Reports}, 19(2):613--626, 2022.
\newblock \href {https://doi.org/10.33048/semi.2022.19.051}
  {\path{doi:10.33048/semi.2022.19.051}}.

\bibitem{Counting_FPT_Delta_corrected}
D.~V. Gribanov, I.~A. Shumilov, and D.~S. Malyshev.
\newblock A faster algorithm for counting the integer points number in
  {$\Delta$}-modular polyhedra (corrected version), 2023.
\newblock Version 7, 8 May 2023.
\newblock URL: \url{https://arxiv.org/abs/2110.01732v7}, \href
  {https://arxiv.org/abs/2110.01732v7} {\path{arXiv:2110.01732v7}}.

\bibitem{SparseILP_Gribanov}
Dmitry Gribanov, Ivan Shumilov, Dmitry Malyshev, and Nikolai Zolotykh.
\newblock Faster algorithms for sparse {ILP} and hypergraph
  multi-packing/multi-cover problems.
\newblock {\em Journal of Global Optimization}, 89(4):1033--1067, 2024.
\newblock \href {https://doi.org/10.1007/s10898-024-01379-z}
  {\path{doi:10.1007/s10898-024-01379-z}}.

\bibitem{Parametric_Counting_Grib}
Dmitry~V. Gribanov, Dmitry~S. Malyshev, Panos~M. Pardalos, and Nikolai~{\relax
  Yu}. Zolotykh.
\newblock A new and faster representation for counting integer points in
  parametric polyhedra.
\newblock {\em Computational Optimization and Applications}, 92(3):811--861,
  2025.
\newblock \href {https://doi.org/10.1007/s10589-024-00632-1}
  {\path{doi:10.1007/s10589-024-00632-1}}.

\bibitem{CountingFixedM}
Dmitry~V. Gribanov and N.~{\relax Yu}. Zolotykh.
\newblock On lattice point counting in {$\Delta$}-modular polyhedra.
\newblock {\em Optimization Letters}, 16(7):1991--2018, 2022.
\newblock \href {https://doi.org/10.1007/s11590-021-01744-x}
  {\path{doi:10.1007/s11590-021-01744-x}}.

\bibitem{counting_Lasserre_book}
Jean-Bernard Lasserre.
\newblock {\em Linear and integer programming vs linear integration and
  counting: a duality viewpoint}.
\newblock Springer Series in Operations Research and Financial Engineering.
  Springer, New York, 2009.
\newblock \href {https://doi.org/10.1007/978-0-387-09414-4}
  {\path{doi:10.1007/978-0-387-09414-4}}.

\bibitem{EvasiveSets_SudakovTomon}
Benny Sudakov and Istv{\'a}n Tomon.
\newblock Evasive sets, covering by subspaces, and point--hyperplane
  incidences.
\newblock {\em Discrete \& Computational Geometry}, 72(3):1333--1347, 2024.
\newblock \href {https://doi.org/10.1007/s00454-023-00601-1}
  {\path{doi:10.1007/s00454-023-00601-1}}.

\bibitem{BlackBoxPolynomials_ZeronDeLoera}
Eduardo~S. Zeron and Jes{\'u}s~A. {De Loera}.
\newblock Computing the degree of black-box polynomials, with applications.
\newblock In Bharath Sriraman, editor, {\em Handbook of Visual, Experimental
  and Computational Mathematics}, pages 777--792. Springer, Cham, 2026.
\newblock \href {https://doi.org/10.1007/978-3-032-16368-4_54}
  {\path{doi:10.1007/978-3-032-16368-4_54}}.

\end{thebibliography}

\end{document}